\documentclass[10pt,twocolumn]{IEEEtran}

\usepackage{amsmath, graphics,amssymb,epsfig,subfigure,color,amsthm,cite}

\newtheorem{theorem}{Theorem}

\begin{document}

\sloppy

\title{Coded Compressive Sensing:\\A Compute-and-Recover Approach}

\author{
  \IEEEauthorblockN{Namyoon Lee and Song-Nam Hong\\}
\thanks{N. Lee  is with Intel Labs,
2200 Mission College Blvd, Santa Clara, CA 95054, USA (e-mail:namyoon.lee@gmail.com). S.-N. Hong  is with Ericsson Silicon Valley Research Lab. 300 Holger Way,
San Jose, CA 95134 (e-mail:sunny7955@gmail.com) }
}

\maketitle

\begin{abstract}
In this paper, we propose \textit{coded compressive sensing} that recovers an $n$-dimensional integer sparse signal vector from a noisy and quantized measurement vector whose dimension $m$ is far-fewer than $n$. The core idea of coded compressive sensing is to construct a linear sensing matrix whose columns consist of lattice codes. We present a two-stage decoding method named \textit{compute-and-recover} to detect the sparse signal from the noisy and quantized measurements. In the first stage, we transform such measurements into noiseless finite-field measurements using the linearity of lattice codewords. In the second stage, syndrome decoding is applied over the finite-field to reconstruct the sparse signal vector. A sufficient condition of a perfect recovery is derived. Our theoretical result demonstrates an interplay among the quantization level $p$, the sparsity level $k$, the signal dimension $n$, and the number of measurements $m$ for the perfect recovery. Considering 1-bit compressive sensing as a special case, we show that the proposed algorithm empirically outperforms an existing greedy recovery algorithm.

\end{abstract}


\vspace{-0.4cm}
\section{Introduction}
Compressive sensing (CS)\cite{ CandesRombergTao2006,Donoho} is a promising technique that recovers a high-dimensional signal represented by a few non-zero elements using far-fewer measurements than the signal dimension. This technique has immense applications ranging from image compression to sensing systems requiring lower power consumption. The mathematical heart of CS is to solve a under-determined linear system of equations by harnessing an inherent sparse structure in the signal.

Let ${\bf x}\in\mathbb{R}^n$ and ${\bf \Phi}\in\mathbb{R}^{m\times n}$ be a real-valued spare signal vector and a compressive sensing matrix that linearly projects a high-dimensional signal in $\mathbb{R}^n$ to a low-dimensional signal in $\mathbb{R}^m$ where $m<n$, respectively. Formally, the noiseless CS problem is to reconstruct sparse signal vector ${\bf x}$ by solving the following $\ell_0$-minimization problem:
\begin{align}
\min \|{\bf x}\|_0~~ {\rm subject}~{\rm to}~~ {\bf y}={\bf \Phi}{\bf x}, \label{eq:L0_CS}
\end{align}
where the collection of non-zero elements' positions in ${\bf x}$, ${\rm supp}({\bf x})$, is defined as $\mathcal{T}=\{i\mid  {\bf x}_i\neq 0,  i\in [1:n]\}$, with cardinality $|\mathcal{T}|=k$. Unfortunately, computational complexity for solving this problem is NP-hard, implying that, in practice, it is computationally infeasible to obtain the optimal solution when $n$ is very large. There exists many practical algorithms that perfectly  reconstruct the sparse signal with polynomial time computational complexity, provided that a measurement matrix has a good incoherence property. Greedy sparse signal recovery algorithms \cite{TroppGilbert2007,Needell2010_CoSaMP,Dai2009_SP,Namyoon2014} became popular due to the computational efficiency in implementing these algorithms.

In practice, obtaining the measurement vector ${\bf y}$ with infinite precision is infeasible. This is because, in many image sensors or communication systems, signal acquisition is performed by using analog-to-digital converters (ADCs) that quantizes each measurement to a predefined value with a finite number of bits. This quantization process makes difficulty in recovering sparse signals, as it might give rise to significant measurement errors, especially when the number of quantization bits is small. Numerous sparse signal recovery algorithms with quantized measurements in \cite{Sun_Goyal_2009,Dai_2009,BoufounosBaraniuk2008,BIHT2013,BP_onebit,Haupt_Baraniuk_2011} were proposed to overcome the impact of the quantization errors. In particular, under the premise that each measurement is quantized with just \textit{one-bit} (i.e., an extreme case of quantization errors), a compressive sensing problem was introduced in \cite{BoufounosBaraniuk2008}. For given ${\bf \Phi}$ and ${\bf x}$, the measurements are obtained using their signs as
\begin{align}
{\bf y} ={\rm sign}\left({\bf \Phi}{\bf x}\right),
\end{align}
where the measurement vector is in the Boolean cube, i.e., ${\bf y}\in \{-1,1\}^m$. It was shown that, with the one-bit measurements, sparse signal vectors with unit-norm can be recovered with high probability by convex optimization techniques \cite{BP_onebit} or iterative greedy algorithms \cite{BIHT2013}.

In this paper, we study a generalized compressive sensing problem in which each measurement is quantized with $p$ levels where $p$ ($\geq 2$) is a prime number. We also consider a quantized source signal, i.e., the non-zero elements of a sparse signal are chosen from a set of integer values, i.e., ${\bf x}\in \mathbb{Z}^n_p$. Such setting can be found in  many applications. For instance, in a random access wireless system, active users among all users in the system send quadrature amplitude modulated symbols (i.e., $p$-level quantized signals) to a receiver, and it detects the active users' signals using a $p$-level ADC.

A fundamental question we ask in this paper is: what is the sufficient condition for the perfect recovery of the integer sparse signal with $p$-level per measurement in the presence of Gaussian noise? To shed light on the answer to this question, we develop a new sparse signal recovery framework, which is referred to as ``\textit{coded compressive sensing}.'' The core idea of coded compressive sensing is to exploit both source and channel coding techniques in information theory. The proposed scheme consists of two cascade encoding and decoding phases. The first phase of encoding is the \textit{compression phase}, in which a high dimensional sparse signal vector in ${\mathbb Z}_p^n$ is compressed to a low dimensional signal vector using a parity check matrix of a maximum distance separable (MDS) linear code. The second phase is the \textit{dictionary coding phase}. In this phase, each dictionary vector (each column vector of the parity check matrix) is encoded to a coded dictionary vector by exploiting a (near) capacity-achieving lattice code for a Gaussian channel. We propose  a two-stage decoding method called ``\textit{compute-and-recover}.'' In the first stage of decoding, a linear combination of the encoded dictionary vectors corresponding the non-zero elements in ${\bf x}\in\mathbb{Z}^n_p$ is decoded. We call this as the dictionary equation decoding stage that produces a noise-free measurement vector. Once the dictionary equation is perfectly decoded, in the second stage of decoding, we apply syndrome decoding to the equivalent finite field representation of the dictionary equation for the sparse signal recovery. Using the proposed scheme, we derive a lower bound of the number of measurements for the perfect recovery as a function of important system parameters: the quantization level $p$, the sparsity level $k$, the signal dimension $n$, and the number of measurements $m$. Considering $p=2$ as a special case, we compare the proposed scheme with existing algorithms developed for the one-bit compressive sensing problem \cite{BoufounosBaraniuk2008}. Numerical results show that the proposed scheme outperforms than binary iterative hard thresholding (BITH) \cite{BIHT2013} in a low signal-to-noise ratio (SNR) regime.

\vspace{-0.5cm}
\section{Coded Compressive Sensing Problem}
In this section, we present a coded compressive sensing framework for an integer sparse signal recovery in the presence Gaussian noise.

\vspace{-0.4cm}
\subsection{Signal Model}
We are interested in a sparse signal detection problem from a compressed measurement in the presence of noise. Let ${\bf x}=[x_1,x_2,\ldots,x_n]^{\top} \in \mathbb{Z}_p^{n\times 1}$ be an unknown sparse signal vector whose sparsity level is equal to $k$, i.e., $\|{\bf x}\|_0=k \ll n$. The measurement equation of quantized compressed sensing is given by
\begin{align}
{\bf y}=S_{p}\left({\bf \Phi}{\bf x} +{\bf n}\right), \label{eq:system_eq}
\end{align}
where $S_{p}$ denotes the $p$-level {\em scalar} quantizer that applied component-wise, and ${\bf y}\in \mathbb{R}^{m\times 1}=[y_1,\ldots,y_m]^{\top}$ and ${\bf n}\in \mathbb{R}^{m\times 1}=[n_1,\ldots,n_m]^{\top}$ denote the measurement and noise vector, respectively. All entries of the noise vector are assumed to be independent and identically distributed (IID) Gaussian random variables with zero mean and variance $\sigma^2/m$, i.e., $n_{i} \sim \mathcal{N}\left(0,\frac{\sigma^2}{m}\right)$ for all $i$.

Our objective is to reliably estimate the unknown sparse signal vector ${\bf x}$ given ${\bf y}$ in the presence of Gaussian noise ${\bf n}$, by appropriately constructing a linear measurement matrix ${\bf \Phi}$ and $p$-level scalar quantizer $S_{p}$. We define a sparse signal recovery decoder $\mathcal{D}:\mathbb{R}^m \rightarrow \mathbb{Z}_p^n$, which maps the measurement vector ${\bf y}$ to an estimate ${\bf \hat x}=\mathcal{D}({\bf y})$ of the original sparse signal vector ${\bf x}$. It is said that the average probability of error is at most $\epsilon>0$ if $\mathbb{P}[{\bf \hat x}\neq {\bf x}] \leq \epsilon$.


\vspace{-0.5cm}
\subsection{Sensing Matrix Construction}
A linear encoding function is represented by a sensing matrix ${\bf \Phi}\in\mathbb{R}^{m\times n}$, which linearly maps the $n$-dimensional sparse vector to an $m$-dimensional output vector, where $m\ll n$. We construct the sensing matrix ${\bf \Phi}$ using the proposed idea, which is referred to as \textit{dictionary coding.}

\subsubsection{Dictionary Basis Vector Selection} %
Let ${\bf \tilde H}\in\mathbb{F}^{{\tilde m}_1\times n}_q$ be a parity check matrix of a $q$-ary $[n,b]$ MDS code, where
$q$ is a prime power $p^s$ for any positive integer $s\in\mathbb{Z}^+$. In this paper we focus on a $q$-ary $[n,b]$ Reed-Solomon (RS) code with field size constraint $q\geq n$. Thus, the parity check matrix of the RS code has ${\tilde m}_1$ full-rank where ${\tilde m}_1=n-b \in \mathbb{Z}^{+}$. The $\ell$th column vector of ${\bf \tilde H}$ is denoted by ${\bf \tilde h}_{\ell}\in \mathbb{F}_{p^s}^{{\tilde m}_1}$ where $\ell\in [1:n]$. We define the one-to-one mapping $h:\mathbb{F}_{p^s}^{{\tilde m}_1} \rightarrow \mathbb{F}_{p}^{{\tilde m}_1s}$ that maps each element of $\mathbb{F}_{p^s}$ into an $s$-length word in $\mathbb{F}_{p}$. For instance, when $p=2$, it is possible to express an element of $\mathbb{F}_{2^s}$ as a binary vector of length $s$. Using this mapping, we can transform each dictionary vector ${\bf \tilde h}_{\ell} \in \mathbb{F}_{p^s}^{{\tilde m}_1}$ into ${\bf h}_{\ell} = h({\bf \tilde h}_{\ell})\in \mathbb{F}_{p}^{{m}_1}$ where $m_1={\tilde m}_1 s$. The transformed column vector ${\bf h}_{\ell}$ is referred to as the $\ell$th dictionary basis vector.

\subsubsection{Dictionary Coding via a Lattice Code}
Dictionary coding is to map a dictionary basis vector in ${\mathbb F}^{m_1}_p$ into a lattice point in $\mathbb{R}^m$ using lattice encoding where $m\geq m_1$.

We commence by providing a brief background for a lattice construction. Let $\mathbb{Z}$ be the ring of Gaussian integers and $p$ be a Gaussian prime. Let us denote the addition over $\mathbb{F}_{p}$ by $\bigoplus$, and let $g:\mathbb{F}_{p} \rightarrow \mathbb{R}$ be the natural mapping of $\mathbb{F}_{p} $ onto $\{a:a\in \mathbb{Z}_p\}\subset\mathbb{R}$. We recall the nested lattice code construction given in [14]. Let $\Lambda = \{ {\boldsymbol \lambda} = {\bf T}{\bf z} : {\bf z} \in \mathbb{Z}^m\}$ be a lattice in ${\mathbb R}^m$, with full-rank generator matrix ${\bf T}\in {\mathbb R}^{m\times m}$. Let $\mathcal{C} = \{{\bf c} = {\bf w}{\bf G} : {\bf w} \in \mathbb{F}_{p}^{m_1} \}$ denote a linear code over $\mathbb{F}_{p}$ with block length $m$ and dimension $m_1$, with generator matrix ${\bf G}\in \mathbb{F}_p^{m_1\times m}$ where $m\geq m_1$. The lattice $\Lambda_1$ is defined through ``\textit{construction A}'' (see \cite{Nazer} and references therein) as
\begin{align}
\Lambda_1 =p^{-1}g(\mathcal{C}){\bf T} + \Lambda,
\end{align}
where $g(\mathcal{C})$ is the image of $\mathcal{C}$ under the mapping function $g$. It follows that $\Lambda \subseteq \Lambda_1 \subseteq p^{-1}\Lambda$ is a chain of nested lattices, such that $\left|\frac{\Lambda_1}{\Lambda}\right| = p^{m_1}$ and $\left|\frac{ p^{-1}\Lambda}{\Lambda_1}\right| = p^{(m-m_1)}$.

For a lattice $\Lambda$ and ${\bf r} \in \mathbb{R}^m$, we define the lattice quantizer ${\rm Q}_\Lambda({\bf r})=\arg \min_{{\bf r}\in \Lambda}\|{\bf r}-{\boldsymbol \lambda}\|_2$, the Voronoi region $\mathcal{V}_{\Lambda}=\{{\bf r}\in \mathbb{R}^m :{\rm Q}_\Lambda({\bf r})={\bf 0}\}$ and $[{\bf r}]\mod \Lambda ={\bf r}-{\rm Q}_\Lambda({\bf r})$. For $\Lambda$ and $\Lambda_1$ given above, we define
the lattice code $\mathcal{L}=\Lambda_1 \cup \mathcal{V}_{\Lambda}$ with rate $R = \frac{1}{m} \log |\mathcal{L}| = \frac{m_1}{m}\log(p)$.

Construction A provides an encoding function that maps a dictionary basis vector ${\bf h}_{\ell} \in \mathbb{F}_{p}^{m_1} $ into a codeword in $\mathcal{L}$. Notice that the set $p^{-1}g(\mathcal{C}){\bf T}$ is a system of coset representatives of the cosets of $\Lambda$ in $\Lambda_1$. Hence, the encoding function $f : \mathbb{F}^{m_1}_p \rightarrow \mathcal{L}$ is defined by
\begin{equation}
f({\bf h}_{\ell}) = \left[p^{-1}g({\bf c}_{\ell}){\bf T}\right] \mod \Lambda,
\end{equation}where
\begin{equation}
({\bf c}_{\ell})^{\top} = ({\bf h}_{\ell})^{\top}{\bf G}.
\end{equation}
Consequently, the $\ell$th codeword vector ${\bf t}_{\ell}$ is produced by the encoding function
\begin{align}
{\bf t}_{\ell}=f({\bf h}_{\ell}),
\end{align}
where each dictionary vector is chosen from lattice codewords in the nested lattice codebook $\mathcal{L}$, i.e.,
${\bf t}_{\ell} \in \mathcal{L}$. Using this construction method, we have a linear sensing matrix consisting of $n$ column vectors as
\begin{align}
{\bf \Phi} =\left[{\bf t}_1, {\bf t}_2, \ldots, {\bf t}_n\right].
\end{align}
The average power of each codeword is assumed to be
\begin{equation}\label{eq:power_const}
\frac{1}{m}\mathbb{E}[\|{\bf t}_{\ell}\|_2^2]\leq 1.
\end{equation}
Finally, in this paper, we choose the shaping lattice $\Lambda$ as a {\em cubic} lattice, namely ${\bf T} = \tau{\bf I}$, which enables that a lattice decoding is implemented by a scalar quantizer (see \cite{Hong_Caire2011} for more details). Here, $\tau$ is chosen to satisfy the power constraint in (\ref{eq:power_const}) as $\tau = \sqrt{8} \mbox{ for } p=2 \mbox{ and } \tau=\sqrt{12} \mbox{ for } p \geq 3$. Then, the element-wise SNR is defined as
\begin{align}
{\rm SNR} = \frac{\tau^2}{\sigma^2/m}.
\end{align}

\subsection{Proposed Scalar Quantizer}

We propose a $p$-level scalar quantizer called {\em sawtooth transform} as depicted in Fig.~\ref{fig:1}, which can be implemented by the modulo operation followed by the scalar quantization as
\begin{equation}\label{eq:quantizer}
S_{p}(\cdot) = \left[Q_{(\tau/p)\mathbb{Z}}(\cdot)\right] \mod \tau\mathbb{Z}.
\end{equation}

\section{Main Result}

In this section, we characterize the sufficient condition for the exact recovery of an integer sparse signal vector. The following theorem is the main result of this paper.

\begin{theorem}\label{th1}
The proposed coded compressive sensing method perfectly reconstructs the sparse signal vector ${\bf x}\in \mathbb{Z}^{n}_p$ with $\|{\bf x}\|_0\leq k$, with vanishing error probability for large enough $n$, provided that
\begin{align}
m\geq \frac{2k\log_p{n}}{1- H_p(Z)},
\end{align}where $H_p(\cdot)$ represents a $p$-ary entropy function and $Z$ denotes an effective quantized noise obtained from the $p$-level quantizer as
$Z = g^{-1}\left(\left[Q_{\mathbb{Z}}\left( N \right)\right] \mod p\mathbb{Z} \right)$,
where $N\sim \mathcal{N}\left(0,\frac{p^2}{{\rm SNR}}\right)$.
\end{theorem}

\begin{proof}

The proof of this theorem is based on the proposed two-stage decoding method called ``\textit{compute-and-recover}''. In the first stage, we decode an integer linear combination of coded dictionary vectors by removing noise, which essentially yields a finite-field sparse signal recovery problem. In the second stage, we apply syndrome decoding over the finite-field to reconstruct the sparse signal vector.

\vspace{-0.28cm}
\subsection{Step 1: Computation of Dictionary Equation}
In this stage, we decode a noise-free measurement vector ${\bf \tilde y} \in \mathbb{F}^{m_1}_p$ from ${\bf y} \in \mathbb{R}^{m}$ using the key property of a lattice code. Recall that dictionary vector is a lattice code; thereby, any integer-linear combination of lattice codewords is again a lattice codeword
\cite{Nazer}. Thus we have that $[\sum_{\ell\in \mathcal{T}}{\bf t}_{\ell}x_{\ell}] \mod \Lambda \in \mathcal{L} $ due to ${\bf x} \in \mathbb{Z}_p^{n \times 1}$. We will first exploit this fact to decode a noise-free measurement vector.

Letting $\mathcal{T}$ be the support set of ${\bf x}$, the noisy measurement vector with the $p$-level quantizer is given by
\begin{align}
{\bf y}&=S_{p}\left(\sum_{\ell\in \mathcal{T}}{\bf t}_{\ell}x_{\ell} +{\bf n}\right) \nonumber \\
&=\left[Q_{(\tau/p)\mathbb{Z}}\left(\sum_{\ell\in \mathcal{T}}{\bf t}_{\ell}x_{\ell} +{\bf n}\right)\right] \mod \tau\mathbb{Z},
\end{align} where the second equality follows from (\ref{eq:quantizer}).

We transform this noisy and quantized measurement into a noiseless finite-field measurement as follows. From the quantized sequence ${\bf y}$, we produce the sequence ${\bf \hat y} \in \mathbb{F}^{m}_p$ with components
\begin{align}
{\hat y}_{i} &= g^{-1}\left(\frac{p}{\tau}{\bf y} \right)\nonumber \\
&=g^{-1}\left(\left[Q_{\mathbb{Z}}\left(\frac{p}{\tau}\left(\sum_{\ell \in \mathcal{T}}t_{\ell,i} + n_{i} \right)   \right)   \right] \mod p\mathbb{Z}\right),
\end{align} for $i=1,\ldots,m$. Since $\frac{p}{\tau}t_{\ell,i} \in \mathbb{Z}$ by construction, and using the obvious identity $Q_{\mathbb{Z}}(v+\zeta) = v+Q_{\mathbb{Z}}(\zeta)$ with $v \in \mathbb{Z}$ and $\zeta \in \mathcal{R}$, we arrive at
\begin{equation}
{\bf \hat y} = \left(\bigoplus_{\ell \in \mathcal{T}}{\bf c}_{\ell}g^{-1}(x_{\ell}) \right) \oplus {\bf z},
\end{equation} where the elements of the discrete additive noise vector ${\bf z}$ are given by
\begin{equation}
z_{i} = g^{-1}\left(\left[Q_{\mathbb{Z}}\left(\frac{p}{\tau}n_{i}\right)\right] \mod p\mathbb{Z}\right),
\end{equation}for $i=1,\ldots,m$.
Since, by linearity, ${\bf v}=\bigoplus_{\ell \in \mathcal{T}}{\bf c}_{\ell}g^{-1}(x_{\ell})$ is a codeword of $\mathcal{C}$, the above channel can be considered as a point-to-point channel with discrete additive noise over $\mathbb{F}_{p}$. Then, we can reliably decode ${\bf v}$ if
\begin{equation}
\frac{m_1}{m}\leq 1 - H_p(Z),
\end{equation} This is an immediate consequence of the well-known fact that linear codes achieve the capacity of symmetric discrete memoryless channel\cite{Dobrushin}. From this result, we can obtain that the sufficient condition for the perfect recovery of the noise-free measurement vector is
\begin{align}
m\geq \frac{m_1}{1 - H_p(Z)}. \label{eq:lattice_rate}
\end{align}

\vspace{-0.3cm}
\subsection{Step 2: Recovery via Syndrome Decoding}

Recall that, in the first stage, the decoder has recovered the dictionary equation, i.e., ${\bf v}=\bigoplus_{\ell \in \mathcal{T}}{\bf c}_{\ell}{\tilde x}_{\ell}$ where ${\tilde x}_{\ell} = g^{-1}(x_{\ell})$. Using the linearity of code $\mathcal{C}$, we have:
\begin{align}
{\bf \tilde y}=\bigoplus_{\ell\in \mathcal{T}}{\bf h}_{\ell}{\tilde x}_{\ell},
\end{align}
where ${\bf \tilde y}$ represents the effective measurement vector in $\mathbb{F}^{m_1}_p$. As a result, the measurement equation can be equivalently rewritten in a matrix form over $\mathbb{F}_{p}$ as
\begin{align}
{\bf \tilde y} ={\bf H}{\bf \tilde x}.
\end{align}
where ${\bf H}\in\mathbb{F}_p^{m_1\times n}$ denotes the effective sensing matrix whose column vectors are selected from dictionary basis vectors ${\bf h}_{\ell}\in\mathbb{F}^{m_1\times 1}_p$.

We would like to recover ${\bf x} = g({\bf \tilde x})\in \mathbb{Z}_p^n$ from the effective measurement vector ${\bf \tilde y}={\bf H}{\bf \tilde x}\in \mathbb{F}_p^{m_1}$ in a noiseless setting and using one-to-one mapping $g(\cdot)$. Unlike the sparse recovery algorithm in a finite field in \cite{Draper}, we apply a syndrome decoding method \cite{Ancheta,Vishwanath2013}. Syndrome decoding harnesses the fact that there is a bijection mapping between a sparse signal (error) vector ${\bf x}$ and the effective measurement (syndrome) vector ${\bf \tilde y}$, provided that ${\bf x}$ contains at most $\lfloor\frac{d_{\rm min}}{2} \rfloor$ non-zero entries, i.e, $k\leq \lfloor\frac{d_{\rm min}}{2} \rfloor$. Recall that, in our construction, the $\ell$th dictionary vector ${\bf h}_{\ell} \in \mathbb{F}_p^{m}$ in ${\bf H}$ was generated from the mapping $h: \mathbb{F}_{p^s}^{{\tilde m}_1}\rightarrow \mathbb{F}_{p}^{m_1}$ where $m_1={\tilde m}_1s$, i.e., ${\bf h}_{\ell}=h({\bf \tilde h}_{\ell})$. Since $h$ is bijection, applying the inverse mapping function ${\bf \bar y}=h^{-1}({\bf \tilde y})\in \mathbb{F}_{p^s}^{{\tilde m}_1}$, we obtain the resultant measurement equation over $\mathbb{F}_{p^s}$ as
\begin{align}
{\bf \bar y} &=h^{-1}\left({\bf H}{\bf \tilde x}\right) =h^{-1}\left({\bf H}\right){\bf \tilde x}= {\bf \tilde  H}{\bf \tilde x},
\end{align}
where the second equality follows from $\mathbb{F}_p^n \subset \mathbb{F}_{p^s}^n$ and the last equality is due to the one-to-one mapping between ${\bf H}$ and ${\bf \tilde H}$ by $h(\cdot)$. Since ${\bf \tilde H}$ was selected from the parity-check matrix of the $p^s$-ary $[n,b]$-RS code whose minimum distance, $d_{\rm min}$ achieves a singleton bound, i.e., $d_{\rm min}= n-b+1={\tilde m}_1+1$. As a result, the syndrome decoding method allows us to recover the sparse signal perfectly, provided that
\begin{align}
k\leq \left\lfloor\frac{\tilde{m}_1+1}{2} \right\rfloor. \label{eq:syndrome_decoding}
\end{align}
Putting two inequalities in \eqref{eq:lattice_rate} and \eqref{eq:syndrome_decoding} together and using the fact $m_1={\tilde m}_1s$ and $s= \log_{p}(n)$, the number of required measurements for the sparse signal recover in the presence of Gaussian noise boils down to
\begin{align}
m\geq \frac{2k\log_p{n}}{1 - H_p(Z)},
\end{align}
which completes the proof.
\end{proof}

\textbf{Remark 1 (Decoding complexity)}:
The proposed two-stage decoding method can be  implemented with a polynomial time computational complexity. In the first stage, the lattice equation can be efficiently decoded with the $p$-level scalar quantizer in \cite{Hong_Caire2011} and the successive decoding algorithm of the polar code \cite{Arikan}, which essentially uses $\mathcal{O}(m\log(m))$ operations. Syndrome decoding used in the second stage can be implemented with polynomial time computational complexity algorithms such as Berlekamp-Massey algorithm, which requires $\mathcal{O}(nk)$ operations in $\mathbb{F}_{p^s}$. Considering $\mathbb{F}_{p^s}$ is a vector space over $\mathbb{F}_{p}$, this amount corresponds to $\mathcal{O}(nks^2)$ operations in $\mathbb{F}_{p}$. Since $n>m$ and $s= \log_{p}(n)$,  the overall computational complexity of the proposed method is at most $\mathcal{O}\left(nk\log^2 (n)\right)$ operations for recovery.

\textbf{Remark 2 (Universality of the measurement matrix)}:
The proposed coded compressive sensing method is universal, as it is possible to recover all $k$ sparse signals using a fixed sensing matrix ${\bf \Phi}$. This universality is practically important, because one may needs to randomly  construct a new measurement matrix ${\bf \Phi}$ for each signal. Some existing one-bit compressive sensing algorithms \cite{BIHT2013, Haupt_Baraniuk_2011,BP_onebit
} do not hold the universality property.

\textbf{Remark 3 (Non-integer sparse signal case)}:
One potential concern for our integer sparse signal setting  is that a sparse signal can have real value components in some applications. This concern can be resolved by exploiting an integer-forcing technique in which ${\bf x}\in\mathbb{R}^n$ is quantized into an integer vector ${\bf  x}_{{\rm int}}\in\mathbb{Z}^n$ and interpreting the residual ${\bf \Phi}({\bf x}-{\bf  x}_{{\rm int}})$ as additional noise. Then, the effective measurements are obtained as
\begin{align}
{\bf y}=S_{p}\left({\bf \Phi}{\bf  x}_{{\rm int}} +{\bf \tilde n}\right), \label{eq:system_eq2}
\end{align}
where ${\bf \tilde n}={\bf n}+{\bf \Phi}({\bf x}-{\bf  x}_{{\rm int}})$ denotes effective noise. Utilizing this modified equation, we are able to apply the proposed coded compressive sensing method to estimate the integer approximation ${\bf  x}_{{\rm int}}$. Assuming the non-zero values in ${\bf x}$ are bounded as  $|{\bf x}_i| \leq U$ for some $U\in\mathbb{R}^{+}$, we conjecture that the proposed scheme guarantees to recover the sparse signal with a bounded estimation error $\|{\bf x}-{\bf  x}_{{\rm int}}\|_2^2\leq k\frac{U^2}{p^2}$ with an increased number of measurements than that in Theorem 1. The rigorous proof of this conjecture will be provided in our journal version \cite{Lee_Hong_2016}.

\textbf{Remark 4 (Noiseless one-bit compressive sensing)}: One interesting scenario is that when a one-bit quantizer and a binary signal are used. In the case of noise-free, the number of required measurements for the perfect recovery is lower bounded by
\begin{align}
m\geq 2k\log_2{n}.
\end{align}

\begin{figure}
\centering \vspace{-0.3cm}
\includegraphics[width=3.4in]{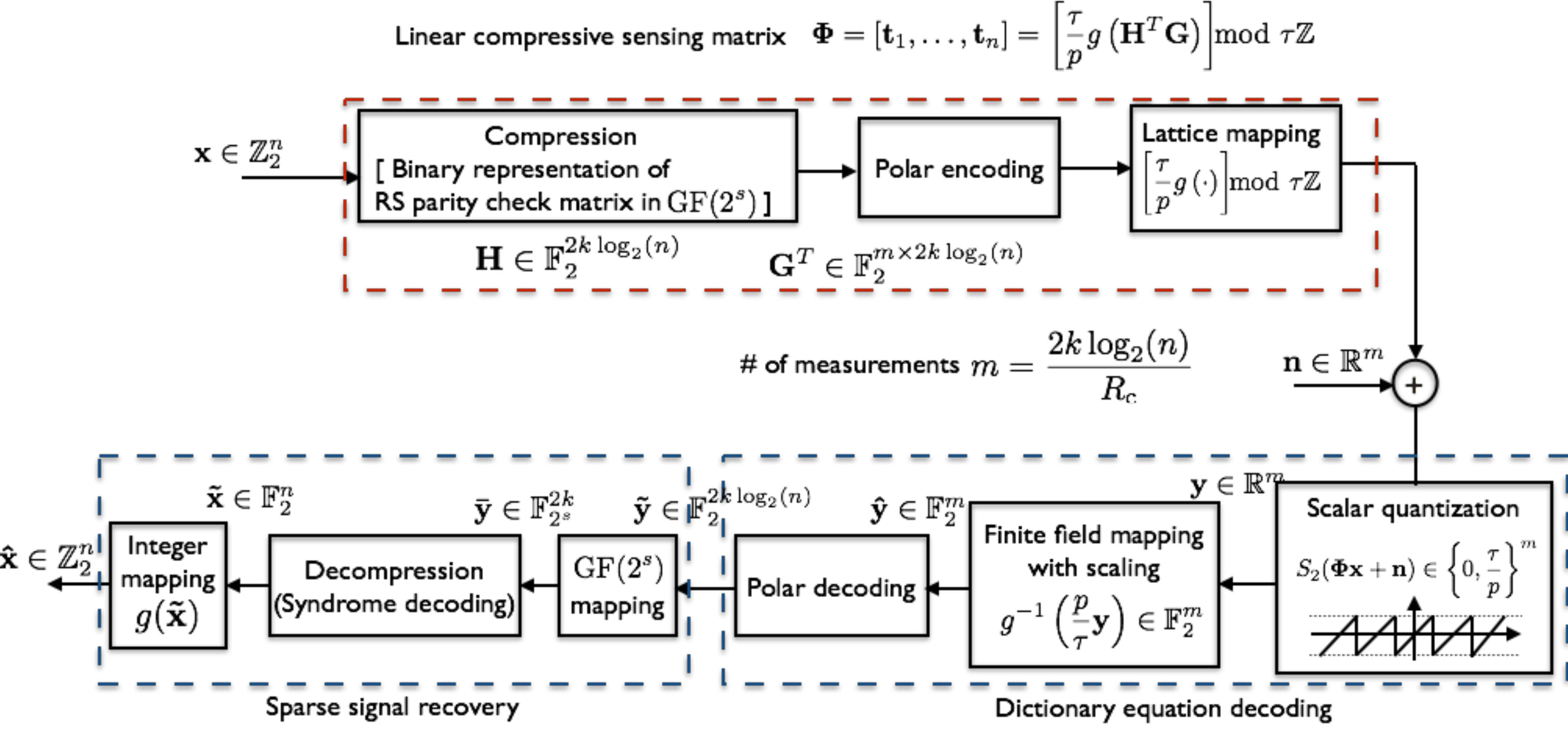} \vspace{-0.2cm}\caption{ The proposed coded compressive sensing framework for the binary sparse signal vector ${\bf x} \in \{0,1\}^n$ with one-bit and noisy measurements.} \label{fig:1} \vspace{-0.5cm}
\end{figure}

 \vspace{-0.2cm}
\section{Numerical Example }
In this section, we provide the signal recovery performance of the proposed coded compressive sensing method for $p=2$, i.e., one-bit compressive sensing, by numerical experiments.

To test the proposed algorithm, $k$-sparse binary vector ${\bf x}\in \mathbb{Z}_2^{511}$ is generated in which the non-zero positions of ${\bf x}$ is uniformly distributed between 1 and 511. A fixed binary sensing matrix ${\bf \Phi}\in \mathbb{F}_2^{180 \times 511}$ is designed by the concatenation of compression matrix ${\bf H}\in \mathbb{F}_2^{90 \times 511}$ and the generator matrix ${\bf G}^T\in \mathbb{F}_2^{180 \times 90}$ of polar code (which is completely determined by the rate-one Arikan's kernal matrix and the information set \cite{Arikan}), as illustrated in Fig. \ref{fig:1}. In particular, the binary compression matrix ${\bf H}$ is obtained from ${\bf \tilde H}$ that is the parity check matrix of ${\rm GF}(2^9)$-ary $[511, 501]$ RS code with the minimum distance of $10$. Therefore, it is perfectly able to perform syndrome decoding up to the sparsity level of 5 in a noiseless case. In addition, we pick the binary polar generator matrix ${\bf G}^T\in \mathbb{F}_2^{180\times 90}$  of code rate $\frac{1}{2}$. We evaluate the perfect recovery probability, i.e., $\mathbb{E}[{\bf 1}\left({\bf x}={\bf \hat x}\right)]$ of the sparse signal  in the presence of noise with variance $\sigma^2$ when the proposed algorithm is applied.

We compare our coded compressive sensing algorithm with the following two well-known one-bit compressive sensing algorithms with some modification for a binary signal.
\begin{itemize}
\item Convex optimization: a variant of the $\ell_1$-minimization method proposed in \cite{BP_onebit} for a binary sparse signal, which is summarized in Table \ref{table1};
\!\!\item Binary iterative hard thresholding (BIHT): a heuristic algorithm in \cite{BIHT2013} with some modifications for the binary signal recovery as in step 3) and 4) of Table \ref{table1}.
\end{itemize}
For the two modified reference algorithms, we use a Gaussian sensing matrix ${\bf \Phi}_{{\rm G}}\in \mathbb{R}^{180 \times 511}$ whose elements are drawn from IID Gaussian distribution $\mathcal{N}(0,\frac{1}{m})$. For each setting of $m$, $n$, $k$, and $\sigma^2$, we perform the recovery experiment for 500 independent trials, and compute the average of perfect recovery rate.

\begin{table}
\center
\caption{A Convex Optimization Algorithm for Binary Sparse Signal}\vspace{-0.2cm}
\begin{tabular}{|l|}
\hline\hline
1)  Initialization:\\
\hspace{5mm} Given $k$, $m$, $n$, $\sigma^2$, ${\bf \Phi}_{{\rm G}}$, ${\bf y}$, and $\mathcal{T}:=\{\emptyset\}$\\
\hline
2) Find $\mathbf{\hat x}$ solving the following convex optimization problem:\\
\hspace{30mm}$\min \|{\bf x}\|_1$ \\
\hspace{15mm}${\rm subject}~{\rm to}~~ ({\bf \Phi}_{{\rm G}}{\bf x})_i {\bf y}_i \geq 0 ~~{\rm for}~~ i\in\{1,\ldots,m\},$\\
\hspace{30mm}$  \sum_{i=1}^m({\bf \Phi}_{{\rm G}}{\bf x})_i {\bf y}_i=m$,\\
\hspace{30mm}$  {\bf x}\geq {\bf 0}$.\\
3) Select the $k$ largest index in $\mathbf{\hat x}$:\\
\hspace{10mm}  ${\mathcal T}=: \arg\max_{|{\mathcal T}|=K} \left\{|{\bf \hat x}|\right\}$.\\
4) Binary signal assignment in ${\mathcal T}$:\\
\hspace{10mm}  ${\bf \hat x}_{\mathcal T}=: {\bf 1}$ and ${\bf \hat x}_{{\mathcal T}^c}=: {\bf 0}$.\\
\hline\hline
\end{tabular}\label{table1}
\end{table}

\begin{figure}
\centering \vspace{-0.2cm}
\includegraphics[width=3.1in]{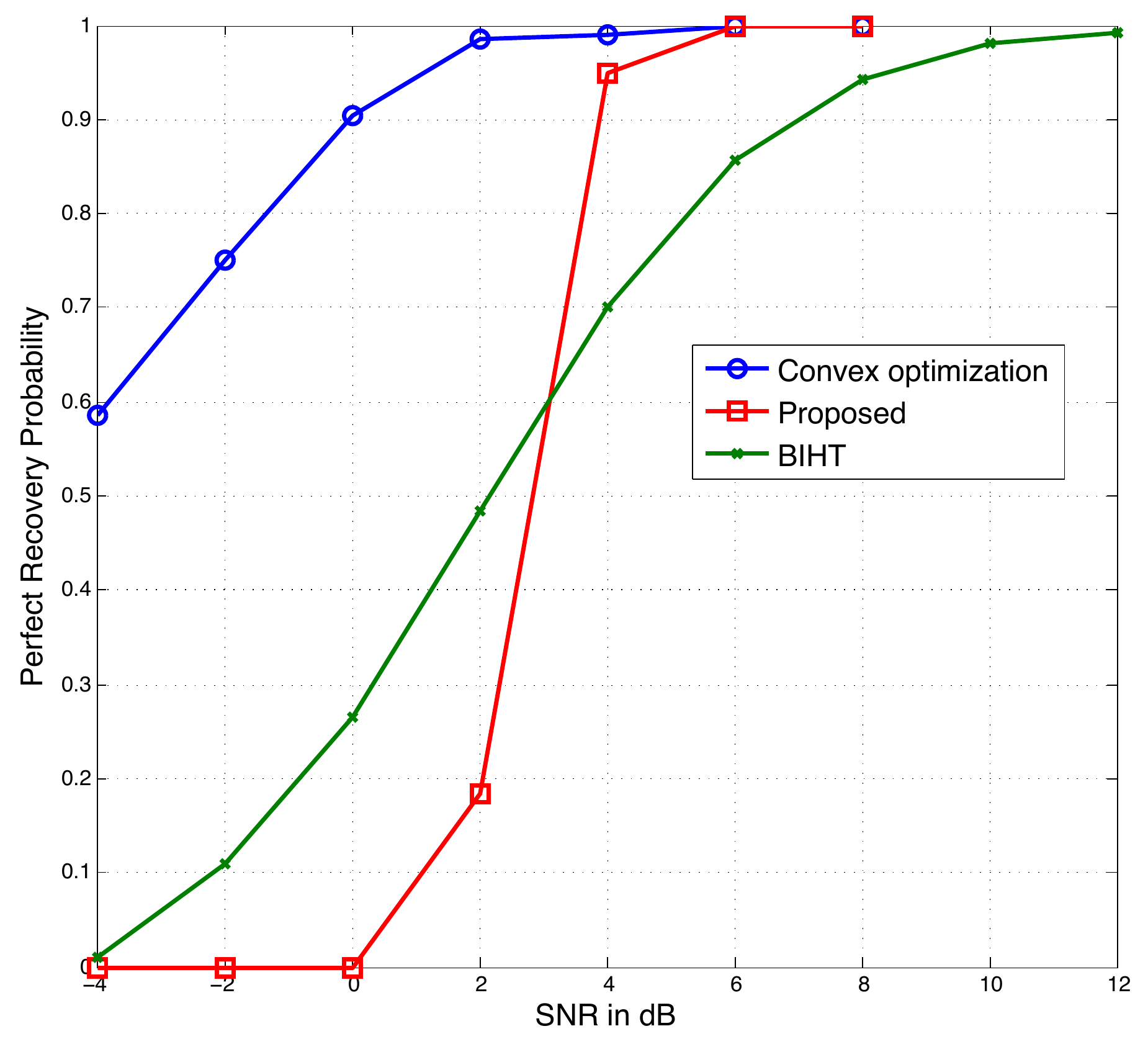} \vspace{-0.4cm}\caption{ Coded one-bit compressive sensing for the binary  sparse signal vector ${\bf x} \in \{0,1\}^n$.} \label{fig:2} \vspace{-0.4cm}
\end{figure}

Fig. \ref{fig:2} plots the perfect recovery probability versus SNR for each algorithm, when $n=511$, $m = 180$, and $k=5$. As can be seen in Fig. \ref{fig:2}, the proposed method outperforms BIHT significantly in terms of the perfect signal recovery performance. Specifically, BIHT is not capable of recovering the signal with high probability until SNR=12 dB, because there are a lot of sign flips in the measurements due to noise. Whereas the proposed algorithm is robust to noise; thereby it recovers the signal with probability one when SNR is 6 dB above. The convex optimization approach provides a better performance than the other algorithms; yet, it requires the computational complexity order of $\mathcal{O}(m^2n^3)$, which is much higher than that of the proposed one.

\vspace{-0.1cm}
 \section{Conclusion}
In this paper, we proposed a novel compressive sensing framework with noisy and quantized measurements for integer sparse signals. With this framework we derived the sufficient condition of the perfect recovery as a function of important system parameters. Considering one-bit compressive sensing as a special case, we demonstrated that the proposed algorithm empirically outperforms  the existing greedy recovery algorithm.

\end{document}